 \documentclass[10pt,draftcls,onecolumn]{IEEEtran}

\usepackage[margin=1in]{geometry}
\usepackage{graphicx} % Required for inserting images
\usepackage{amsthm}
\usepackage{amsmath}
\usepackage{hyperref}

\newtheorem{theorem}{Theorem}
\newtheorem{lemma}[theorem]{Lemma}

\newtheorem{dfn}[theorem]{Definition}

\newtheorem{assump}[theorem]{Assumption}

\usepackage{amsfonts}
\usepackage{xcolor}
\newcommand{\p}{\mathbb{P}}
\newcommand{\rin}{R_{\mathrm{in}}}
\newcommand{\xtil}{\widetilde{x}}
\newcommand{\Eprime}{E'}
\newcommand{\lambdatil}{\widetilde{\lambda}}
\newcommand{\Lambdatil}{\widetilde{\Lambda}}

\title{Error Exponents for DNA Storage Codes \\ with a Variable Number of Reads}

\author{ Yan Hao Ling, Nir Weinberger, and Jonathan Scarlett \thanks{Y.H.~Ling is with the Department of Computer Science, National University of Singapore.  (Email: \url{lingyh@nus.edu.sg})}
\thanks{N.~Weinberger is with the Faculty of Electrical and Computer Engineering, Technion. (Email:  \url{nirwein@technion.ac.il}) }
\thanks{J.~Scarlett is with the Department of Computer Science, Department of Mathematics, and the Institute of Data  Science, National University of Singapore. (Email:  \url{scarlett@comp.nus.edu.sg})}
\thanks{The research of J.S.~was supported by the Singapore National Research Foundation (NRF) under its AI Visiting Professorship programme.  The research of N.W.~was partially supported by the Israel
Science Foundation (ISF), grant no.~1782/22.} }
\date{\today}
\begin{document}
\maketitle

\begin{abstract}
    In this paper, we study error exponents for a concatataned coding based class of DNA storage codes in which the number of reads performed can be variable.  That is, the decoder can sequentially perform reads and choose whether to output the final decision or take more reads, and we are interested in minimizing the average number of reads performed rather than a fixed pre-specified value.  We show that this flexibility leads to a considerable reduction in the error probability compared to a fixed number of reads, not only in terms of constants in the error exponent but also in the scaling laws.  This is shown via an achievability result for a suitably-designed protocol, and in certain parameter regimes we additionally establish a matching converse that holds for all protocols within a broader concatenated coding based class.
\end{abstract}

\section{Introduction}

In recent years, significant research attention has been paid to characterizing the capacity of DNA storage systems; see \cite{shomorony2022information} for a recent overview.  While capacity is perhaps the most fundamental notion of a performance limit, it is limited by its lack of consideration of \emph{how quickly} the error probability approaches zero as the problem size increases.  This limitation is addressed by studies of \emph{error exponents}, which characterize the exponential decay rate of the error probability at rates below capacity.  Error exponents have long been studied for channel coding problems \cite{gallager1968information,csiszar2011information}, and were recently studied for DNA storage in \cite{merhav,Weinberger,ling2024exact}.

In more detail, \cite{merhav} studied error exponents of random codes, whereas \cite{Weinberger,ling2024exact} focused on a class of \emph{concatenated codes} in which an an ``inner code'' for the sequencing channel is used in a black-box manner for individual molecules, and an ``outer code'' is used to handle the entire set of molecules.  This class of codes is practically well-motivated, since the inner code can make use of well-established techniques from point-to-point channel coding.  Moreover, in \cite{Weinberger,ling2024exact}  the sequencing channel need not be a discrete memoryless channel as was assumed in \cite{merhav}; for example, insertions and deletions can be considered.

While the results of \cite{ling2024exact} (building on \cite{Weinberger}) provide exact (i.e., optimal) error exponents for a class of concatenated codes under a standard setup, they turn out to leave significant room for improvement under a minor modification of the setup.  In particular, the error exponent is dominated by the event that there are unusually few distinct molecules sampled, and correspondingly, their converse result holds \emph{even when the sequencing error probability is zero}.  This suggests that the performance could be significantly improved by allowing a \emph{variable number of reads}, so that more reads can be performed when needed, without increasing the \emph{average} number of reads too much.

In this paper, we make this intuition precise by studying error exponents of concatenated codes under a variable number of reads, where we now consider the \emph{average} number of reads instead of a pre-specified number.  We show that this setup permits a considerably smaller error probability, and we attain matching achievability and converse results in certain regimes of interest.

\subsection{Contributions}

We introduce our problem setup and notation in Section \ref{sec:setup}, but for the purpose of outlining our contributions, we note that we will consider codewords containing $M$ molecules of length $L$, let $N$ denote the number of reads performed, and let $p$ denote the sequencing error probability.  Our contributions are:
\begin{itemize}
    \item We follow the setup of \cite{Weinberger,ling2024exact} with concatenated coding and an index-based inner code (see Section \ref{sec:setup}).  
    Under the standard scaling choices $L = \Theta(\log M)$, $\mathbb{E}[N] = \Theta(M)$, and $p = o(1)$, we show that the error probability with a variable number of reads decays as $e^{-\Theta(M\log\frac{1}{p})}$, thus being significantly smaller than $e^{-\Theta(M)}$ with a fixed number of reads.  This result holds for all coding rates that are achievable via the coding techniques in \cite{Weinberger,ling2024exact} (our techniques in this paper will have the same set of achievable rates).
    \item We establish an achievable error exponent (i.e., specifying the hidden constant in $\Theta(\cdot)$ notation in $e^{-\Theta(M\log\frac{1}{p})}$), and in certain parameter regimes we additionally give a matching converse, thus establishing the exact error exponent in the variable-reads setting.  
    \item Our bounds hold under suitably-defined adversarial models of sequencing errors, and we consider two kinds of adversary: A ``strong adversary'' that knows the decoder and can see into the future (e.g., knowing exactly when sequencing errors will occur), and a ``weak adversary'' that only knows the message and codebooks, and must act solely based on the past.  Our achievability result (Theorem \ref{thm:achievability}) holds even under the strong adversary, and our converse is proved first for the strong adversary (Theorem \ref{thm:converse}) and then for the weak adversary (Theorem \ref{thm:converse2}).
\end{itemize}
Before proceeding to the problem setup, we make some remarks to help set the context for our results within the existing literature:
\begin{itemize}
    \item The class of codes that we consider is motivated by previous practical and theoretical uses of concatenated codes for DNA storage (e.g., \cite{lenz2020achievable,ren2022dna}).  We refer the reader to \cite{Weinberger} for further discussion on the practical motivation.
    \item While the concatenated codes that we consider are powerful and well-motivated, we note that they can be suboptimal, precluding certain advanced techniques such as clustering \cite[Sec.~5.1]{shomorony2022information} and ``full'' random coding \cite{merhav}. In particular, the exponent is only positive for rates up to a certain threshold that can be strictly worse than the one in \cite{merhav} (see \eqref{eq:rate_achieved} below).
    \item Similarly to the related works \cite{merhav,Weinberger,ling2024exact}, our work is complementary to the extensive work on \emph{coding-theoretic} considerations for DNA storage codes (e.g., see \cite{lenz2019coding,song2020sequence,kovavcevic2018codes} and the references therein), which typically seek distinct goals such as good distance properties.
\end{itemize}

\section{Problem Setup} \label{sec:setup}

We follow the problem setup from \cite{Weinberger,ling2024exact}, while allowing a variable number of reads.  For completness, we provide the full description.

\subsection{Basic Notation} 

We adopt a standard information-theoretic model of DNA storage, with the building blocks being as follows:
\begin{itemize}
    \item The goal is to reliably store and recover a \emph{message} $m$ drawn uniformly at random from a finite set of messages.
    \item A \emph{molecule} is a length-$L$ string whose symbols come from some finite alphabet $\mathcal{X}$ (e.g., $\mathcal{X} = \{A,C,G,T\}$).
    \item Each message is associated with an unordered collection of $M$ input molecules that are stored.  This collection is referred to as the (outer) \emph{codeword} for $m$.
    \item At the decoder, $N$ \emph{reads} are performed.  For each read, (i) one of the $M$ input molecules is \emph{sampled} uniformly at random; and (ii) the molecule in $\mathcal{X}^L$ is passed through a \emph{sequencing channel} (independently of all other reads) and the decoder sees the output, whose alphabet we denote by $\mathcal{Y}^{(L)}$.  For example, the channel might be a discrete memoryless channel, but we do not require this to be the case, so $\mathcal{Y}^{(L)}$ need not be an $L$-fold product.  
    \item In the fixed-length setup, the number of reads $N$ is pre-specified, whereas in the variable-length setup, the decoder may choose when to stop performing reads and thus $N$ is a random variable.
    \item The decoder uses the reads to produce an estimate $\hat{m}$ of the message.  
\end{itemize}
We proceed to provide the full details and introduce aspects that are more specific to our setup.

\subsection{Concatenated Coding Based Class of Protocols}

We now describe the concatenated coding based class of DNA storage codes that we consider.  The following definition summarizes the relevant notions of inner codes.

\begin{dfn} \label{def:inner}
We define the following:
\begin{enumerate}
    \item An \emph{inner code} $(X_{\rm in},D_{\rm in})$ with parameters $(\rin, L)$ is given by the following:
    \begin{itemize}
        \item A sequence of $\exp(\rin L)$ distinct molecules, each of length $L$, denoted by $X_{\rm in} = (x_1, x_2, \dots, x_{\exp(\rin L)})$. 
        \item An inner decoding function operating on the sequencing channel output, $D_{\rm in} : \mathcal{Y}^{(L)} \rightarrow \{x_1,\ldots, x_{\exp(\rin L)}\}$.
    \end{itemize}
    \item  Given the inner codebook $(X_{\rm in},D_{\rm in})$, for each $x \in X_{\rm in}$, the (inner) \emph{error probability} for $x$ is the probability that $D_{\rm in}(y) \neq x$, where $y$ is distributed according to the sequencing channel with input $x$.  The highest (among all $x \in X_{\rm in}$) of these error probabilities gives the \emph{maximal error probability} of the inner code.
        \item A sequence of inner codebooks $(X_{\rm in}^{(L)}, D_{\rm in}^{(L)})_{L=1}^\infty$ achieves a rate $\rin$ if each $(X_{\rm in}^{(L)}, D_{\rm in}^{(L)})$ has parameters $(\rin, L)$, and the maximal error probability of $(X_{\rm in}^{(L)}, D_{\rm in}^{(L)})$ approaches zero as $L \to \infty$. Such a rate is said to be \emph{achievable} for the sequencing channel.
\end{enumerate}
\end{dfn}

We highlight that we use the inner code in a ``black-box'' manner, only assuming that it has maximal error probability approaching zero as $L$ increases, in accordance with the above notion of achievable (inner) rate.

The class of concatenated codes that we consider is then defined as follows; when specialized to the case of a fixed number of reads, this recovers the class studied in \cite{Weinberger,ling2024exact}.

\begin{dfn}
Given an inner code $(X_{\rm in},D_{\rm in})$, the protocol is said to perform \emph{separate inner and outer coding} if it satisfies the following two properties:
\begin{itemize}
    \item For any message $m$ at the encoder, the resulting input (i.e., the \emph{outer codeword}) is a multiset $A_m$ of size $M$ whose elements are chosen from the inner codebook $X_{\rm in}$.
    \item After performing some number $\tilde{N}$ of reads and observing the outputs $y_1, y_2,\ldots, y_{\tilde{N}}$, any decision made by the decoder (i.e., whether to continue performing reads or not, and if not, what to output as the estimate) depends only on $D_{\rm in}(y_1), D_{\rm in}(y_2),\ldots, D_{\rm in}(y_{\tilde{N}})$.
\end{itemize}
\label{dfn:use_inner}
\end{dfn}

In this paper, we further impose the requirement that the code is \emph{index-based} according to the following definition (which was also adopted in \cite{Weinberger,ling2024exact}).

\begin{dfn}
    A codebook is \emph{index-based} if there exist $M$ disjoint sets of molecules $(B_i)_{i=1}^M$ of equal size, such that every outer codeword $A_m$ contains exactly one molecule from each $B_i$. 
    \label{dfn:index_based}
\end{dfn}

Index-based codes are often considered to be favorable for keeping the encoder and decoder simple; see \cite[Sec.~4.3]{shomorony2022information} and \cite{Weinberger} for further discussion.  We note that Definition \ref{dfn:index_based} ensures that every $A_m$ contains $M$ distinct molecules, meaning that the multiset in Definition \ref{dfn:use_inner} is a set.

\subsection{Sampling, Sequencing, and Decoding}

As noted in \cite{ling2024exact}, under separate inner and outer coding it is useful to think of sequencing and inner decoding as a single step; every time a read is performed, the following occurs:
\begin{itemize}
    \item One of the $M$ (distinct) input molecules $x \in A_m$ is chosen uniformly at random;
    \item The molecule is passed through the sequencing channel;
    \item The inner code's decoder is applied to the channel output to produce $\xtil$.
\end{itemize}
For a fixed inner code, the second and third steps can be viewed as inducing an overall distribution $P(\xtil|x)$, where it we have $P(x|x) = 1-o(1)$ as $L \to \infty$ when the inner code achieves rate $\rin$.
% Note that the second and third of the steps can be summarized by a distribution $P(x'|x)$, and under Definition \ref{dfn:use_inner} it is only $P(x'|x)$ that is relevant, not the intermediate channel output.  When the inner code achieves rate $\rin$, we have $P(x|x) = 1-o(1)$ as $L \to \infty$. 
However, we will not use this probabilistic approach; rather, in accordance with our goal of using the inner code in a ``black-box'' manner, we will instead adopt a \emph{worst-case} (adversarial) model of what happens when the inner code fails.  Specifically:
\begin{itemize}
    \item For each read performed, a \emph{sequencing error} occurs with some probability $p = o(1)$, independently of all other reads.  (Hence, we have $\xtil=x$ with probability $1-p = 1-o(1)$.)
    \item When a sequencing error occurs, we assume that $\xtil$ may equal \emph{any} element of the inner code.  (We allow this element to be $x$ itself, so strictly speaking $p$ is an \emph{upper bound} on the error probability of the inner code.)
\end{itemize}
In the second dot point, we will view $\xtil$ as being chosen by an adversary; see Definitions \ref{def:adversary1}, \ref{def:adversary2}, and \ref{def:adversary3} below.

Our focus in this paper is on the case where the number of reads is variable.  Thus, at each time step, having performed some number of reads, the decoder may choose to perform another one, or may choose to stop and output the estimate of the message.  Thus, the number of reads $N$ is a random variable (even if the decoder is deterministic, these decisions depend on the random message and noise).  We are interested in keeping the \emph{average} number of reads small, and accordingly write
\begin{equation}
    \overline{N} = \mathbb{E}[N]
\end{equation}
to denote this average.  We assume that $\overline{N} < \infty$, meaning that the decoder stop in finite time with probability one.

For a uniformly random message $m$ and its estimate $\hat{m}$, we define the \emph{error probability}
\begin{equation}
    P_e = \p( \hat{m} \ne m ).
\end{equation}
We re-iterate that we are not only interested in attaining $P_e \to 0$, but also in the speed of convergence to zero.

\subsection{Scaling of parameters} \label{sec:scaling}

While in principle there any many possibilities for how $(L,M,\overline{N})$ scale with respect to one another, we will focus our attention on scaling regimes that are the most widely-adopted in information-theoretic studies (e.g., \cite{Weinberger,ling2024exact,merhav,fundamental_limit_dna_2017}), and are practically well-motivated (e.g., $L = \Theta(\log M)$ corresponding to relatively short reads). 
Specifically, we focus on the case that the following quantities are constant as $M \to \infty$:
\begin{itemize}
    \item $c = \overline{N} / M$, representing the (average) coverage depth;
    \item $\beta = \frac{L}{\log M}$, characterizing the molecule length;
    \item $\rin$, which is the inner codebook rate as per Definition \ref{def:inner};
    \item $R$, which we write for the outer code rate such that there are $\exp(RML)$ messages.
\end{itemize}
Observe that the inner code contains $\exp(\rin L) = M^{\beta \rin}$ molecules, and accordingly we define
\begin{equation}
\alpha = {\beta \rin}
    \label{eq:alpha_def}
\end{equation} 
so that this simplifies to $M^\alpha$. We assume that $\alpha > 1$, since the use of index-based codes (Definition \ref{dfn:index_based}) requires having at least $M$ inner codewords.\footnote{Moreover, for $\alpha \le 1$ the capacity, as defined in the next paragraph, is zero even when the code need not be index-based.}

The \emph{capacity} is defined as the supremum of $R$ for which there exists a sequence of codes (indexed by $M$) attaining $P_e \to 0$.  
Under a simpler model in which the $M$ molecules are observed directly in a uniformly random order with no sequencing errors, the capacity is given by \cite[Lemma 1]{fundamental_limit_dna_2017}\footnote{The result in \cite{fundamental_limit_dna_2017} considers a binary code where $\rin = 1$, but the proof can easily be adapted to obtain this more general version.}
\begin{equation}
    \rin - \frac{1}{\beta} = \frac{\alpha-1}{\beta}.
\end{equation}
Following \cite{ling2024exact}, we express the outer rate of the protocol as a fraction of the capacity:
\begin{equation}
    R_0 = \frac{R}{\rin - \frac{1}{\beta}} = \frac{R \beta}{\alpha - 1}.
    \label{eq:rbar_def}
\end{equation}
Thus, the number of possible messages that the encoder can receive is 
\begin{equation}
    \exp(RML) = \exp((\alpha-1)R_0 M\log M). \label{eq:num_messages}
\end{equation}
% \begin{equation}
%     \exp(RML) = \exp((\alpha-1)\rbar M\log M). \label{eq:RML}
% \end{equation}
For the remainder of the paper, we will treat all of $(c,\rin,R,R_0,\alpha,\beta)$ as constants. 

% \subsection{Variable-Sample Setting}

% Follow the model above. Codebooks are assumed to be index-based, where a codeword is an element of $\prod_{i=1}^M B_i$. There are $\exp((\alpha-1)R_0 M \log M)$ possible codewords. A sequencing error occurs with probability $p= o(1)$. When a sequencing error occurs, the molecule received by the decoder may be any element of the inner code, possibly including the correct molecule (hence, strictly speaking, $p$ is an \emph{upper bound} on the sequencing error probability). 

% The decoder samples adaptively. More precisely, at time $t$, the decoder can choose to sample another molecule $y_t$, or stop sampling and guess the message $m$. We are concerned about two parameters: (1) $P_e$, probability of making an error, and (2) $N$, the expected stopping time.

\subsection{Discussion on Capacity with a Fixed vs.~Variable Number of Reads}

We briefly pause to discuss the overall storage capacity (for general codes that need not be concatenated nor index-based).  This capacity can be defined under a fixed number of reads with parameter $c = N/M$, or under a variable number of reads with parameter $c = \overline{N}/M$, and it is not obvious whether these two capacity values (with all other parameters being identical) are always the same.  Clearly, allowing a variable number of reads cannot decrease capacity. 

Building on ideas from variable-rate channel coding \cite[Thm.~10]{verdu2010variable}, we can also provide a partial statement in the opposite direction by showing that any achievable rate in the variable-$N$ setting implies a certain positive result for the fixed-$N$ setting.  Consider any variable-$N$ code performing $\overline{N}$ reads on average and satisfying $P_e \to 0$.  By Markov's inequality, we have $N \le \overline{N}(1+\delta)$ with probability at least $1 - \frac{1}{1+\delta} = \frac{\delta}{1+\delta}$.  Then, we can form a fixed-length code that performs $\overline{N}(1+\delta)$ reads, produces the same output as the variable-$N$ decoder whenever the latter stops by time $\overline{N}(1+\delta)$, and declares an error otherwise.  This code performs a fixed number $\overline{N}(1+\delta)$ of reads, and succeeds with probability at least $\frac{\delta}{1+\delta} - o(1)$.

Since $\delta$ is arbitrarily small, this argument implies that the following conditions would suffice for the fixed-$N$ and variable-$N$ capacities to be identical:
\begin{itemize}
    \item[(i)] The \emph{strong converse} holds in the fixed-$N$ setting;\footnote{That is, the \emph{$\epsilon$-capacity} defined with respect to the requirement $P_e \le \epsilon + o(1)$ (rather than the usual $P_e \to 0$) is identical for all fixed $\epsilon \in (0,1)$.}
    \item[(ii)] The capacity is continuous with respect to the parameter $c = N/M$ (in the fixed-$N$ setting) or $c = \overline{N}/M$ (in the variable-$N$ setting).
\end{itemize}
These conditions seem reasonable to expect, but formally establishing them may be difficult.  Condition (i) is complicated by the fact that  existing upper bounds on capacity typically rely on Fano's inequality (e.g., see \cite{shomorony2022information,merhav}), which only gives a weak converse.\footnote{An alternative to Condition (i) would be a condition that $N$ is \emph{sufficiently concentrated} around its mean for the optimal variable-$N$ code, so that the above-mentioned success probability bound $\frac{\delta}{1+\delta} - o(1)$ becomes $1-o(1)$.  However, such a concentration result also seems difficult to prove in general.}
Condition (ii) is partially supported by existing capacity \emph{bounds} being continuous in $c$ (again see \cite{shomorony2022information,merhav}), but in general we do not have a single-letter expression for the capacity.  

Perhaps unsurprisingly given the above discussion, the \emph{achievable rates} in our work (i.e., the rates under which we attain $P_e \to 0$) will be identical to those attained in the fixed-$N$ regime with index-based concatenated codes in \cite{Weinberger,ling2024exact}.  (See \eqref{eq:rate_achieved} below.)

\subsection{Discussion on Random vs.~Adversarial Sequencing Errors}

There are several possible models for what happens when a sequencing error occurs, with the main distinction being random vs.~adversarial (both of which are explored in \cite{ling2024exact}).  For example, a simple randomized model would be that in which each sequencing error leads to a uniformly random incorrect molecule.  In this paper, we focus entirely on adversarial models for the sequencing errors; compared to a randomized model, this amounts to a strength in the achievability results, but a limitation in the converse results.  We believe that adversarial models are partially justified by the preference for using the inner code in a ``black-box'' manner without needing to know its inner workings, which points to ``worst-case'' thinking regarding what happens when errors occur.

In the setting of a fixed number of reads, random and adversarial sequencing error models give identical error exponents \cite{ling2024exact}, at least under the standard scaling laws adopted in our work.  Intuitively, this is because sequencing errors are not dominant, and instead the dominant error event is failing to sample enough distinct molecules.  With that no longer being the case under a variable number of reads, it may be that random and adversarial models behave more differently; this direction is left for future work.

\section{Main Results} \label{sec:results}

Here we present and discuss our main results, deferring the proofs to Sections \ref{sec:pf_ach} and \ref{sec:pf_conv}.

\subsection{Strong and Weak Adversaries}

Recall that for each molecule sampled, a sequencing error occurs with probability $p$, and when this is the case, the decoded molecule may be arbitrary.  We consider two adversarial models for how this molecule may be chosen, one corresponding to a stronger adversary, and one corresponding to a weaker one.  Both our achievability and converse results will ultimately hold under both adversaries, but it will be convenient to first prove the converse for the stronger adversary.

We first state some basic assumptions that hold under both adversaries.

\begin{assump} \label{def:adversary1}
(Basic Assumptions) We adopt the following assumptions:
\begin{itemize}
    \item The message $m$ is chosen uniformly at random, and is revealed to the adversary.
    \item The protocol uses separate inner and outer coding (Definition \ref{dfn:use_inner}), and the code is index-based (Definition \ref{dfn:index_based}).
    \item The adversary knows the inner and outer codebooks (and thus the encoder).
    \item The decoder is deterministic.
    \item Each read has a sequencing error with probability $p$ independently of all other reads.
    \item Whenever a sequencing error occurs, the molecule is replaced by some molecule in the inner codebook (possibly including the correct one).
    % \item The adversary never changes the index, but can change the molecule to any other molecule with the same index. {\bf \color{blue} [I think this doesn't need to be an assumption, but rather an adversary design choice.]}
\end{itemize}
\end{assump}

We believe that the first two and last two assumptions are natural in an adversarial setting.  The restriction to concatenated codes is justified by the fact although they can be suboptimal in terms of capacity (let alone error exponent), they are often favorable in terms of complexity.  The restriction to \emph{index-based} codes is also partly justified by the goal of reducing complexity, though ideally a converse could be proved (in future work) without needing this assumption, and similarly for the assumption of deterministic decoding. 
 Our current converse proofs rely on both of these assumptions holding.

Next, we define the additional assumptions for our stronger adversary.

\begin{assump} \label{def:adversary2}
(Strong Adversary) For our strong adversary, we adopt the following additional assumptions:
\begin{itemize}
    \item The adversary knows the (deterministic) decoder.
    \item The adversary knows the infinite-length sequence of molecules that would be sampled (before sequencing errors) if the decoder were to continue performing reads.
    \item The adversary also knows the infinite-length sequence of which reads will be subject to sequencing errors.
\end{itemize}
\end{assump}

While the first of these assumptions is fairly natural, the second and third may be less so, and may appear to ``give the adversary too much power''.  This amounts to a strength for an achievability result, but a weakness for a converse result.  Fortunately, we will also establish our converse under the following relaxed assumptions.

\begin{assump} \label{def:adversary3}
(Weak Adversary) For our weak adversary, we relax the assumption stated in Assumption \ref{def:adversary2} as follows:
\begin{itemize}
    \item The adversary does not know the decoder.
    \item If a sequencing error occurs at the $i$-th read, then the adversary's choice of replaced molecule is only allowed to depend on what was observed in previous reads (i.e., those indexed by $1,\dotsc,i-1$).
    \item The adversary does not know which future reads will be subject to sequencing errors.
\end{itemize}
\end{assump}

We now proceed to state our main results.

\subsection{Statement of Main Results}

Our main achievability result is stated as follows, and holds even under Assumptions \ref{def:adversary1} and \ref{def:adversary2} (which, for achievability, is a stronger result than one using Assumptions \ref{def:adversary1} and \ref{def:adversary3}).

\begin{theorem} \label{thm:achievability}
    Let $c = \log \frac{1}{1-R_0-\delta}$ for any $\delta \in (0,1-R_0)$.  Under Assumptions \ref{def:adversary1} and \ref{def:adversary2}, there exists a protocol with error probability $P_e \le p^{(\delta + o(1)) M}$ and an average number of reads $\overline{N} \le (c + o(1))\cdot M$ when the message is chosen uniformly at random.\footnote{In fact, our proof shows the same for the worst-case message, not only a random message.}
\end{theorem}

Next, we state our converse result under the strong adversary, and then under the weak adversary.  The latter result of course implies the former, but we find it useful to prove the former one first.

\begin{theorem}
    Let $c = \log \frac{1}{1-R_0-\delta}$ with $\delta$ being sufficiently small such that $\delta < ce^{-c}$. Under Assumptions \ref{def:adversary1} and \ref{def:adversary2} (strong adversary), any decoder with error probability satisfying $P_e \le o(p^{\delta M+1})$ must require at least $\overline{N} \ge (c-o(1))M$ reads in expectation when the message is chosen uniformly at random.
    \label{thm:converse}
\end{theorem}

\begin{theorem}
    Let $c = \log \frac{1}{1-R_0-\delta}$ with $\delta$ being sufficiently small such that $\delta < ce^{-c}$. Under Assumptions \ref{def:adversary1} and \ref{def:adversary3} (weak adversary), any decoder with error probability satisfying $P_e \le o(2^{-M}p^{\delta M+1})$ must require at least $\overline{N} \ge  (c-o(1))M$ reads in expectation when the message is chosen uniformly at random.
    \label{thm:converse2}
\end{theorem}

We now proceed to compare the achievability and converse.

\subsection{Discussion and Numerical Evaluation}

We first note that the statements $P_e \le p^{(\delta + o(1)) M}$, $P_e \le o(p^{\delta M+1})$, and $P_e \le o(2^{-M}p^{\delta M+1})$, while slightly different, all have right-hand sides behaving as $p^{(\delta + o(1)) M} = \exp\big( -\big(\delta M \log \frac{1}{p}\big)(1+o(1)) \big)$ when $p = o(1)$, and thus all amount to having the same error exponent (i.e., the same value of $\lim_{M \to \infty} -\frac{1}{M\log(1/p))} \log P_e$, which in this case equals $\delta$).

Recall that error exponents for concatenated codes with a fixed number of reads were studied in \cite{Weinberger,ling2024exact}, with the achievability parts using index-based codes and only assuming that $p = o(1)$.  
One might be tempted to compare those exponents with Theorem \ref{thm:achievability}, but in fact, we can make a stronger claim than a higher exponent: We also improve the \emph{scaling law} in the exponent, reducing the error probability from $e^{-\Theta(M)}$ to $e^{-\Theta(M\log(1/p))}$.  Thus, using a variable number of reads can provably allow considerably smaller error probability.  As discussed in the introduction, the intuition is that being variable-length allows us to avoid certain (dominant) ``bad events'' where too few distinct molecules are sampled.

Next, we compare our achievability and converse bounds.  As discussed above, the error exponents themselves are identical, but with the important difference that the converse is restricted to $\delta < ce^{-c}$.  This means that we have matching achievability and converse bounds, but only for a limited range of parameters.  As we see in Figure \ref{fig:examples}, the range where the converse holds is limited but still reasonable, and corresponds to regimes where the average number of reads is not too high.

Finally, observe that $P_e \le p^{(\delta + o(1)) M}$ implies $P_e \to 0$ whenever $\delta > 0$.  Thus, we attain a positive error exponent in Theorem \ref{thm:achievability} whenever $c > \log\frac{1}{1-R_0}$, or equivalently $R_0 < 1-e^{-c}$.  Combining this with the definition of $R_0$ in \eqref{eq:rbar_def}, and choosing $R_{\rm in}$ to be arbitrarily close to the sequencing channel capacity $C_{\rm in}$, we find that we attain a positive error exponent for all rates satisfying
\begin{equation}
    R < (1-e^{-c}) \Big(C_{\rm in} - \frac{1}{\beta}\Big). \label{eq:rate_achieved}
\end{equation}
This matches the set of rates with a positive error exponent in the related prior works studying the fixed-$N$ setting \cite{Weinberger,ling2024exact}, so our improvements are only in the decay rate of $P_e$ and not in the set of rates for which $P_e \to 0$.

    \begin{figure}
        \begin{centering}
            \includegraphics[width=0.6\columnwidth]{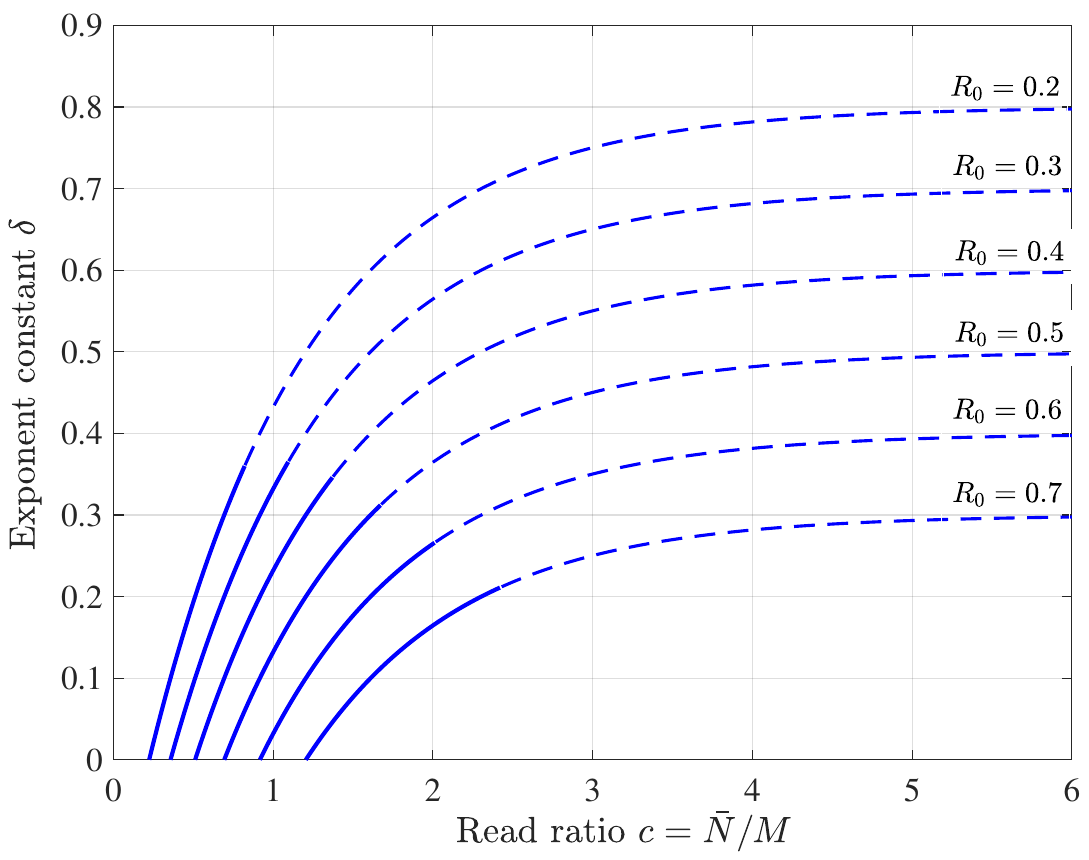} 
            \par
        \end{centering}
        
        \caption{Achievable error exponents for various $R_0$ values (dashed), and a matching converse in the low-$\delta$ regime (solid).} \label{fig:examples}
    \end{figure}

\section{Achievability Proof} \label{sec:pf_ach}

In this section, we prove Theorem \ref{thm:achievability}.  We first describe the codebook choice and the decoding rule, the latter of which makes use of the parameter $\delta \in (0,1-R_0)$ from the theorem statement:
\begin{itemize}
    \item \underline{Codebook}: We choose an index-based codebook $\{A_m\}$ such that, for any $\epsilon > 0$, when $M$ is sufficiently large, it holds that 
    \begin{equation}
        |A_i \cap A_j| < (R_0 + \epsilon) M, \quad \forall i \ne j, \label{eq:intersect_cond}
    \end{equation}
    with $R_0$ being as in \eqref{eq:rbar_def}.  The existence of such a codebook was proved in \cite[Lemma 7]{ling2024exact} using a simple greedy construction analogous to the classical Gilbert-Varshamov construction (e.g., see \cite[Ex.~5.19]{gallager1968information}).
    \item \underline{Decoding}: We say that a codeword $A_i$ is \emph{$\delta$-consistent} with the decoded molecules if there are at most $\delta M$ decoded molecules not in $A_i$.  The decoding rule is then as follows: \emph{Continue sampling until exactly one codeword $A_i$ is $\delta$-consistent, and declare that the message is $i$.}  (If a point is ever reached where none of the $A_i$ are $\delta$-consistent, the decoder stops and declares failure.)
\end{itemize}
The following lemma characterizes the average number of samples taken.

\begin{lemma} \label{lem:ach_time}
    Under the preceding codebook and decoding rule, with sequencing error probability $p = o(1)$, the expected number of reads is at most $(\log\frac{1}{1-R_0-\delta} + o(1))\cdot M$.
\end{lemma}
\begin{proof}
    Fix an arbitrarily small constant $\epsilon > 0$ and suppose that $M$ is large enough such that \eqref{eq:intersect_cond} holds.  
    We first observe that once we see at least $(R_0+\epsilon+\delta)M$ decoded molecules in $A_i$, the decoding process stops and outputs $i$.  This is because once this is true, we have from \eqref{eq:intersect_cond} that all of the other messages $j \ne i$ must have more than $\delta M$ decoded molecules outside $A_j$.

    Letting $m$ be the true message, it follows that the average stopping time is upper bounded by the average time to seeing $(R_0+\epsilon+\delta)M$ decoded molecules in $A_m$.  To characterize this, note that after seeing $k$ distinct decoded molecules from $A_m$, every new sample gives a probability at least $(1-p)(1-\frac{k}{M})$ of seeing the $(k+1)$-th one.  The average time for this to occur is at most $\frac{1}{(1-p)(1-\frac{k}{M})}$, since ${\rm Geometric}(q)$ has mean $\frac{1}{q}$.

    The expected time taken to see $(R_0+\epsilon+\delta)M$ molecules from $A_m$ is therefore upper bounded by
    \begin{equation}
        \sum_{k=0}^{(R_0+\epsilon+\delta)M} \frac{1}{(1-p)(1-\frac{k}{M})} = \frac{1}{1-p} M \bigg(\log \frac{1}{1-R_0-\epsilon-\delta} + o(1)\bigg).
    \end{equation}
    Since $\epsilon$ is arbitrarily small and $p = o(1)$, the conclusion follows.
\end{proof}

Next, we study the error probability.

\begin{lemma} \label{lem:ach_perr}
    Under the preceding codebook and decoding rule, and with sequencing error probability $p = o(1)$, the error probability behaves as $P_e \le p^{(\delta + o(1)) M}$.
\end{lemma}
\begin{proof}
    Let $m$ be the true message.  In order for an error to occur upon stopping, it must be that $m$ itself is no longer $\delta$-consistent, which is only possible if there are at least $\delta M$ sequencing errors.  Moreover, as noted in the proof of Lemma \ref{lem:ach_time}, once we see at least $(R_0+\epsilon+\delta)M$ decoded molecules in $A_m$, the decoding process stops and outputs $m$.  Combining these two observations gives: \emph{If we see at least $(R_0+\epsilon+\delta)M$ distinct molecules without sequencing errors before incurring $\delta M$ sequencing errors, then the decoding decision is correct.}

    Suppose that at a certain point we have seen $k$ distinct decoded molecules (some of which may be from sequencing errors). The probability that the next molecule will be a new molecule without sequencing error is $(1-p)(1-\frac{k}{M})$, while the probability that the next sample will incur a sequencing error is $p$. Hence, the probability of incurring a sequencing error before seeing a new molecule without a sequencing error is given by $\frac{p}{p+(1-p)(1-\frac{k}{M})}$.  Observe that for $k \leq (R_0+\epsilon+\delta)M$, this probability is upper bounded as follows:
    \begin{equation}
        \frac{p}{p+(1-p)(1-\frac{k}{M})} \leq \frac{p}{(1-p)(1-\frac{k}{M})} \leq \frac{p}{(1-p)(1-R_0-\epsilon-\delta)}.
        \label{eq:new_molecule}
    \end{equation}
    To understand the implication of this finding on the error probability, we imagine a binary string being formed as we sample:
    \begin{itemize}
        \item Begin with the empty string.
        \item When we see a sequencing error, append a `0' to the string.
        \item When we see a new decoded molecule (i.e., one different from any seen before) without a sequencing error, append a `1' to the string.
        \item The remaining case is that we see an already-seen molecule without a sequencing error, in which case we do not extend the string.
        \item Stop when the string contains $\delta M$ 0s or $(R_0+\epsilon+\delta)M$ 1s.
    \end{itemize}
    Note that by construction, the stopping rule corresponds to having $\delta M$ sequencing errors, or having $(R_0+\epsilon+\delta)M$ distinct molecules without sequencing errors.
    
    %The event that we see $\delta M$ sequencing errors before $(R_0+\epsilon+\delta)M$ molecules directly correspond to binary strings in which $\delta M$ 0s appear before $(R_0+\epsilon+\delta)M$ 1s. 
    By \eqref{eq:new_molecule}, the probability of seeing any specific string with $\delta M$ 0s is upper bounded by
    \begin{equation}
        \left(\frac{p}{(1-p)(1-R_0-\epsilon-\delta)}\right)^{\delta M}.
    \end{equation}
    Moreover, there are at most $2^{\delta M + (R_0+\epsilon+\delta)M}$ strings to consider, as any binary string with length exceeding $\delta M + (R_0+\epsilon+\delta)M$ must either contain at least $\delta M$ 0s or $(R_0+\epsilon+\delta)M$ 1s. 
    Hence, by the union bound, the probability that $\delta M$ 0s appear before $(R_0+\epsilon+\delta)M$ 1s is at most
    \begin{equation}
        2^{\delta M + (R_0+\epsilon+\delta)M} \cdot \left(\frac{p}{(1-p)(1-R_0-\epsilon-\delta)}\right)^{\delta M}.
    \end{equation}
    This simplifies to $p^{(\delta + o(1)) M}$ due to $\epsilon$ being arbitrarily small and $p = o(1)$.
\end{proof}

\section{Converse Proofs} \label{sec:pf_conv}

In this section, we prove Theorems \ref{thm:converse} (strong adversary) and \ref{thm:converse2} (weak adversary).  The structure of the section is as follows:
\begin{itemize}
    \item In Section \ref{sec:further}, we introduce some further notation and definitions.
    \item In Section \ref{sec:pf_stronger}, we state two main lemmas (Lemmas \ref{lem:not_in_s} and \ref{lem:main}) and show how they imply Theorem \ref{thm:converse}.
    \item In Sections \ref{sec:pf_not_in_s} and \ref{sec:pf_prob_f}, we prove the two main lemmas.
    \item In Section \ref{sec:pf_weaker}, we describe the changes in the analysis for proving Theorem \ref{thm:converse2}.
\end{itemize}

\subsection{Further Notation and Definitions} \label{sec:further}

Our analysis will center around the sequence of molecules that have been sampled (before sequencing errors) and observed (after sequencing errors) when the decoder stops and outputs a decision.  Accordingly, we provide the following definition.

\begin{dfn}
    We define the following sequences with (random) length $N$:
    \begin{itemize}
        \item $\Lambda = (\Lambda_1,\dotsc,\Lambda_N)$ denotes the sequence of molecules that have been sampled (before sequencing errors) when the decoder stops;
        \item $\Lambdatil  = (\Lambdatil_1,\dotsc,\Lambdatil_N)$ denotes the sequences of molecules that have been observed by the decoder (after sequencing errors) when the decoder stops.
    \end{itemize}
    The sequence length, written as $N = |\Lambdatil| = |\Lambda|$, is the \emph{stopping time} of the decoder.
\end{dfn}

We similarly write $\lambda$ and $\lambdatil$ (with $i$-th entries $\lambda_i$ and $\lambdatil_i$, and lengths $|\lambda|$ and $|\lambdatil|$) to represent \emph{generic} sequences of molecules, whose lengths need not equal $N$.  Note also that we can view $\Lambda$ as a length-$N$ truncation of an infinite sequence $(\Lambda_1,\Lambda_2,\dotsc)$ whose entries are drawn i.i.d.~via $\Lambda_i \sim {\rm Uniform}(A_m)$ conditioned on the message being $m$.  Recall that the strong adversary in Assumption \ref{def:adversary2} knows this entire infinite-length sequence.

% We say that an event occurs with high probability if, as $M\rightarrow \infty$, the probability of the event occuring approaches 1. Similarly, if the probability of the event approaches 0, we say that it occurs with vanishing probability.

We additionally introduce the following definition specific to index-based codes (see Definition \ref{dfn:index_based}).

\begin{dfn} \label{def:index_seq}
    The (random) \emph{index sequence} is a function $F : \mathbb{Z}^+ \rightarrow {1,2,\ldots, M}$, where $F(j) = i$ indicates that the $j$-th molecule sampled (before sequencing errors) is in $B_i$ (as specified in Definition \ref{dfn:index_based}).  We let $f$ denote a (fixed) realization of such a sequence.
    % generic (fixed) sequence, and let $F$ denote the (random) index sequence induced by the random sampling of molecules.
\end{dfn}

Importantly, the sequence $F$ is independent of the message, since the sampling is uniform over the $M$ molecules, of which there is always one per $B_i$ for index-based codes.  That is, the entries of $F$ are drawn i.i.d.~from ${\rm Uniform}(\{1,\dotsc,M\})$.  
We also mention that while this definition is taken with respect to the molecules before sequencing errors, we will consider an adversary that never changes the index, and hence (for our adversary) the index sequence remains identical after sequencing errors.

Observe also that conditioned on any index sequence $f$ and any message $m$, the infinite-length sequence $(\Lambda_1,\Lambda_2,\dotsc)$ is uniquely determined, since the $j$-th molecule is determined by the $j$-th index and the unique molecule of codeword $m$ corresponding to that index.

The following definition provides a set of index sequences $f$ which will prove to be amenable to lower bounding the error probability (because $\delta M$ sequencing errors will suffice to confuse the decoder).

\begin{dfn} \label{def:setS}
    Fix $c'',R' > 0$, as well as $\delta \in (0,1-R_0)$ from the statements of Theorems \ref{thm:converse} and \ref{thm:converse2}.  We define $S(c'',R')$ to be the set of index sequences $f$ such that we can partition the time indices $\{1,2,\ldots, c''M\}$ into sets $T_1$ and $T_2$ such that:
    \begin{itemize}
        \item[(i)]$|T_1|\leq \delta M$;
        \item[(ii)] There are at most $R'M$ distinct values in $\{f(j)\}_{j \in T_2}$.
    \end{itemize}
\end{dfn}

For each message $m$, let $D_m$ be the decoding region of message $m$, i.e. $\lambdatil \in D_m$ if the decoder, upon receiving the sequence $\lambdatil$, declares that the message is $m$ and stops sampling.\footnote{Hence, if some $\lambdatil$ is such that the decoder doesn't stop, then it appears in none of the $D_m$ sets (or alternatively, we can view such $\lambdatil$ as being in a set $D_0$ representing the decision ``perform another read'').} Observe that we may assume the following:
\begin{equation}
    \text{The set }~\bigcup_{m = 1}^{\exp(RML)} D_m~\text{ is prefix-free} \label{eq:prefix_free}
\end{equation}
(i.e., if $\lambdatil \in D_m$ for some $m$, then no prefix of $\lambdatil$ is in any decoding region). This is because the sampling stops upon a decision to decode, and there is no notion of decoding again after that point.

Let 
\begin{equation}
    E_m = \bigcup_{m'\neq m} D_{m'} \label{eq:Em}
\end{equation}
be the error region, so that if the true message is $m$ but the decoder receives a sequence in $E_m$, the decoder decides on some $m'\neq m$ and there is an error. Moreover, define
\begin{equation}
    \Eprime_m = \{\lambdatil \,:\, \exists \lambda \in E_m ~~\text{s.t.}~~ d_H(\lambda,\lambdatil) \le \delta M \},
\end{equation}
where $d_H(\lambda,\lambdatil)$ is the number of entries in which $\lambda$ and $\lambdatil$ (with $|\lambda|=|\lambdatil|$) have different molecules. That is, for each $\lambdatil \in \Eprime_m$, there exists some $\lambda \in E_m$ of the same length as $\lambdatil$ that differs in at most $\delta M$ molecules.  

\subsection{Proof of Theorem \ref{thm:converse} (Converse Under a Strong Adversary)} \label{sec:pf_stronger}

Recall the set $S(c'',R')$ of index sequences defined in Definition \ref{def:setS}.  The idea of the proof is to show that $f \in S(c'',R')$ with high probability, and then lower bound the conditional error probability given $f \in S(c'',R')$.  Formally, the main intermediate steps are summarized in the following two lemmas, in which we define
\begin{equation}
    c = \log\frac{1}{1-R_0-\delta}
\end{equation}
in accordance with the theorem statement (which also specifies $\delta \in (0,1-R_0)$), with $R_0$ being given in \eqref{eq:rbar_def}.

\begin{lemma}
    Given any $c'' \in (0,c)$, there exists $R' < R_0$ such that the index sequence $F$ lies in $S(c'',R')$ with probability $1-o(1)$.
    \label{lem:not_in_s}
\end{lemma}

\begin{lemma}
There exists an adversary satisfying Assumptions \ref{def:adversary1} and \ref{def:adversary2} such that the following holds: 
Suppose that the message $m$ is chosen uniformly at random, and let $f$ be an index sequence in $S(c'',R')$ for some $c'' < c$ and $R' < R_0$. Conditioned on the index sequence $f$ occurring, we have
\begin{equation}
    \p(\Lambdatil \in E_m | f) \geq p^{\delta M+1} \Big(\p(|\Lambdatil| \leq c''M|f) - p- \exp(-(R_0-R')(\alpha-1)M\log M)\Big).
\end{equation}
\label{lem:main}
\end{lemma}

Note that here and subsequently, we write ``$f$'' in probabilities as a shorthand for the event ``$F=f$'', and similarly for other quantities such as $m$.   
We defer the proofs of these lemmas to subsequent subsections, and first show how they imply Theorem \ref{thm:converse}.

\begin{proof}[Proof of Theorem \ref{thm:converse}]
Consider a strategy that samples at most $c'M$ molecules in expectation for some $c' < c$, i.e., $\overline{N} = \mathbb{E}(|\Lambdatil|) \le c' M$.  Let $c''$ be between $c'$ and $c$. By Markov's inequality, we have
\begin{equation}
    \p(|\Lambdatil|\leq c''M) \geq 1 - \frac{\mathbb{E}(|\Lambdatil|)}{c''M} \geq 1-\frac{c'}{c''}. \label{eq:apply_markov}
\end{equation}
Moreover, we have the following lower bound on the error probability:
\begin{align}
    &\p(\Lambdatil \in E_m) \nonumber \\
    & \geq \p(\Lambdatil \in E_m \land F \in S(c'',R')) \\
    & = \sum_{f \in S(c'',R')} \p(f) \cdot \p(\Lambdatil \in E_m|f )\\
    & \geq \sum_{f \in S(c'',R')} \p(f) \cdot p^{\delta M+1} \Big(\p(|\Lambdatil| \leq c''M|f) -p - \exp((R'-R_0)(\alpha-1)M\log M)\Big) \label{eq:lb_Em_step3}\\
    & \geq  p^{\delta M+1} \Big(\p(|\Lambdatil| \leq c''M \land F \in S(c'',R')) - o(1)\Big), \label{eq:lb_Em_final}
\end{align}
where \eqref{eq:lb_Em_step3} follows from Lemma \ref{lem:main}, and \eqref{eq:lb_Em_final} applies $p = o(1)$, $R' < R_0$, and $\alpha > 1$.

Moreover, Lemma \ref{lem:not_in_s} states that $\p(F \notin S(c'', R')) = o(1)$, and combining this with \eqref{eq:apply_markov} gives
\begin{equation}
    \p(|\Lambdatil| \leq c''M \land F \in S(c'',R')) \geq \p(|\Lambdatil|\leq c''M) - \p(F \notin S(c'',R')) \ge \frac{c'}{c''} - o(1).
\end{equation}
Substituting into \eqref{eq:lb_Em_final}, we obtain 
\begin{equation}
     p^{\delta M+1} \Big( \p(|\Lambdatil| \leq c''M \land F \in S(c'',R')) - o(1)\Big) \geq p^{\delta M+1}(c'/c'' - o(1)) = \Omega(p^{\delta M+1}).
\end{equation}
Recalling that this lower bound is a consequence of $\overline{N} \le c' M$ for some $c' < c$, taking the contrapositive gives the statement of Theorem \ref{thm:converse}.

\end{proof}

\subsection{Proof of Lemma \ref{lem:not_in_s} (High Probability of $F \in S$)} \label{sec:pf_not_in_s}

Recall that $c = \log\frac{1}{1-R_0-\delta}$, and hence $R_0 = 1-\delta - e^{-c}$.  Since $c'' < c$ and $\delta < ce^{-c}$, and since the function $f(z) = ze^{-z} + e^{-z}$ is decreasing for $z \ge 0$, we conclude that there must exist some $R' < R_0$ satisfying the following:
\begin{gather}
    1-e^{-c''} < R' + \delta, \label{eq:c_ineq1} \\
    1-e^{-c''}-c''\cdot e^{-c''} < R'. \label{eq:c_ineq2}
\end{gather}
Recalling that $F$ is the random index sequence (Definition \ref{def:index_seq}), we define two useful random variables depending on $F$: 
\begin{itemize}
    \item Let $Z$ denote be the number of distinct values among $F(1), F(2),\ldots, F(c''M)$;
    \item Let $Z_1$ be the number of values that appear exactly once in the sequence $F(1), F(2),\ldots, F(c''M)$.
\end{itemize}
We claim that if $Z-\min(Z_1, \delta M) \leq R'M$, then $F \in S(c'', R'M)$ (defined in Definition \ref{def:setS}, which specifies two sets $T_1$ and $T_2$). To see this, let $T_1$ be any $\min(Z_1, \delta M)$ values of $j$ such that $F(j)$ appears exactly once in $F(1),F(2),\ldots, F(c''M)$; then, the number of distinct values in $\{F(i)\}_{i \in T_2}$ is $Z - |T_1| = Z - \min(Z_1, \delta M) \leq R'M$ as required in Definition \ref{def:setS}.
    
Recalling that each $F(j)$ is generated uniformly from $\{1,2,\ldots, M\}$, we can apply McDiarmid's inequality \cite{Bou13} on $Z$.  Specifically, when one of the $F(j)$ changes, the number of unique molecules changes by at most 1, so
\begin{equation}
    \mathbb{P}(|Z - \mathbb{E}(Z)| \geq \Omega(M)) \leq \exp(-\Omega(M)),
\end{equation}
which is exponentially small in $M$.

Have established concentration around $\mathbb{E}(Z)$, we now proceed to study $\mathbb{E}(Z)$ itself.  By linearity of expectation and the fact that each each read has a $1-\frac{1}{M}$ probability of failing to sample a given molecule, we have
    \begin{equation}
        \mathbb{E}(Z) = M(1-(1-1/M)^{c''M}) = (1-e^{-c''})M + o(M), \label{eq:E_X}
    \end{equation}
    and applying $1-e^{-c''} < R' + \delta$ from \eqref{eq:c_ineq1}, we obtain
    \begin{equation}
        \mathbb{P}(Z \geq (R'+\delta)M) \le \mathbb{P}(|Z- \mathbb{E}(Z)| \geq \Omega(M)) = \exp(-\Omega(M)), \label{eq:X_finding}
    \end{equation}
    which is exponentially small.  This shows that $Z- \delta M \leq R' M$ with probability $1-o(1)$.

    We use similar reasoning to show that $Z-Z_1 \leq R'M$ with probability $1-o(1)$. Observe that $Z-Z_1$ is equal to the number of $j$ such that $F(j)$ appears at least twice in $F(1), F(2),\ldots, F(c''M)$, and therefore changing one sample changes the output of $Z-Z_1$ by at most 1, so McDiarmid's inequality gives
    \begin{equation}
        \mathbb{P}(|(Z-Z_1)- \mathbb{E}(Z-Z_1)| \geq \Omega(M)) \leq \exp(-\Omega(M)). \label{eq:XY_conc}
    \end{equation}
    For each molecule, the number of times it appears follows a binomial distribution, and thus the probability of seeing a specific molecule exactly once is
    \begin{equation}
        (c''M)(1/M)(1-1/M)^{c''M-1} = c'' \cdot e^{-c''} + o(1),
    \end{equation}
    which implies
    \begin{equation}
        \mathbb{E}(Z_1) = M c'' \cdot e^{-c''} + o(M). \label{eq:E_Y}
    \end{equation}
    Combining \eqref{eq:E_X} and \eqref{eq:E_Y}, and applying $1-e^{-c''}-c''\cdot e^{-c''} < R'$ from \eqref{eq:c_ineq2}, it follows that
    \begin{equation}
        \mathbb{E}(Z-Z_1) \leq (1-e^{-c''}-c''\cdot e^{-c''})M + o(M) \leq R'M - \Omega(M),
    \end{equation}
    and combining this with \eqref{eq:XY_conc}, we obtain
    \begin{equation}
        \mathbb{P}(Z-Z_1 \geq R'M) \leq \exp(-\Omega(M)). \label{eq:XY_finding}
    \end{equation} 
Finally, combining the finding \eqref{eq:X_finding} on $Z$ and the finding \eqref{eq:XY_finding} on $Z-Z_1$, we conclude that, with probability $1-o(1)$, it holds that
\begin{equation}
    Z - \min(Z_1, \delta M) \leq R'M,
\end{equation}
and thus $F \in S(c'', R')$.

\subsection{Proof of Lemma \ref{lem:main} (Probability Bound Given $f \in S$)} \label{sec:pf_prob_f}

Recall that here we consider the strong adversary satisfying Assumptions \ref{def:adversary1} and \ref{def:adversary2}. %, which includes the adversary knowing the decoding regions $\{D_m\}$.  
We have $f \in S(c'', R')$ by assumption, and we let $T_1, T_2$ satisfy the corresponding conditions in Definition \ref{def:setS} (i.e., $|T_1| \le \delta M$, and there are at most $R'M$ distinct values of $\{f(j)\}_{j \in T_2}$).  Recall also that $F$ is independent of the message and the occurrence of sequencing errors, so conditioning on $F=f$ does not impact their distributions.

For each message $m$, let $\lambda^{(m)}$ be the hypothetical infinite-length sequence of molecules sampled (before sequencing errors) when the true message is $m$; note that given $f$, this $\lambda^{(m)}$ is deterministic.  Let $t_m$ be the corresponding stopping time for message $m$ when there are no sequencing errors, and consider the set
\begin{equation}
    U = \{ m : t_m \leq c''M\}. \label{eq:setU}
\end{equation}
It will be useful to distinguish whether the resulting decoder output at time $t_m$ is equal to $m$ or not. Accordingly, let $\Phi$ be the set of messages such that the decoder outputs $m$ when there are no sequencing errors.  We partition $U$ as follows:
\begin{gather}
    U_1 = U \cap \Phi, \label{eq:defU1} \\
    U_2 = U \setminus \Phi = U \setminus U_1. \label{eq:defU2}
\end{gather}
Moreover, we let $\lambda|_{T}$ denote the subsequence of $\lambda$ corresponding to indices in $T$ (in generic notation).  The adversary's strategy is as follows:
\begin{itemize}
    \item The adversary is only active when the following three conditions all hold:
    \begin{itemize}
        \item[(i)] Sequencing errors occur at every molecule in $T_1$;
        \item[(ii)] There exists some $m' \in U_1$ with $m' \neq m$ and $\lambda^{(m)}|_{T_2} = \lambda^{(m')}|_{T_2}$.
        \item[(iii)] A random variable $\Psi \sim {\rm Bernoulli}(p)$ generated by the adversary (independent of everything else) is equal to $1$. 
    \end{itemize}
    \item If the adversary is inactive, then they always set $\lambdatil_j = \lambda_j$, i.e., the correct molecule is received for every read.  In other words, the adversary ``does nothing''.
    \item If the adversary is active, then given $m' \in U_1$  defined in item (ii) above,\footnote{Note that this is where the adversary uses the knowledge of the decoding regions.  Knowing the decoding regions implies knowing the stopping times $t_m$, and thus knowing the set $U$.} for all $j \in T_1$, the adversary changes the molecules of $(\lambda^{(m)})_{j}$ to become the corresponding molecule $(\lambda^{(m')})_j$.  All other molecules are left unchanged.
\end{itemize}
Note that the choice $\Psi \sim {\rm Bernoulli}(p)$ implies that the adversary is active with probability at most $p$, even conditioned on any specific $f$.  It may appear unusual to have the adversary do nothing with probability $1-p$, but the intuition is that this helps mitigate the impact of the adversary on the \emph{average} number of reads taken.\footnote{This is specifically used in \eqref{eq:not_in_U}--\eqref{eq:combining_U} below.}

We now present two useful lemmas.  Recalling that $E_m$ is the set of decoded molecule sequences that lead to incorrect decoding when the message is $m$ (see \eqref{eq:Em}), we have the following.

\begin{lemma}
    Consider any $m$ and any $m' \in U_1$ with $m' \neq m$, and suppose that $\lambda^{(m)}|_{T_2} = \lambda^{(m')}|_{T_2}$.  Then, conditioned on $m$ being the true message, we have 
    \begin{equation}
        \p(\Lambdatil \in E_m \land \Psi=1|m) \geq p^{\delta M+1}. 
    \end{equation}
    \label{lem:unique_mt2}
\end{lemma}
\begin{proof}
    % As noted above, with probability at least $p^{\delta M}$, the adversary is active and can change the molecules of $(\lambda^{(m)})_{j}$ to $(\lambda^{(m')})_j$.
     With probability $p$ with have $\Psi=1$, and with probability at least $p^{\delta M}$ the required sequencing errors at the molecules indexed by $T_1$ (with $|T_1| \le \delta M$) occur.  Thus, the adversary is active with probability at least $p^{\delta M+1}$.  
     Supposing that the adversary is active, we have the following: 
     Since $m' \in U_1$ (and therefore $t_{m'} \leq c''M$), the first $c''M$ decoded molecules of exactly match $\lambda^{(m')}$ (i.e., those of message $m'$ if there were no sequencing errors), and the decoder will stop and make the decision at time $t_{m'}$, which is at most $c''M$. Since $m' \in U_1$, we have $m' \in \Phi$, which means the output of the decoder is equal to $m' \neq m$.  Due to the prefix-free property in \eqref{eq:prefix_free}, the decoder does not stop at any time earlier than $t_{m'}$.
\end{proof}

\begin{lemma}
    Among all $m\in U_1$, there can only be at most $\exp(R'M(\alpha-1)\log M)$ possible values of $\lambda^{(m)}|_{T_2}$, where $\alpha > 1$ is defined in \eqref{eq:alpha_def}.
    \label{lem:unique_v}
\end{lemma}
\begin{proof}
    Note that for all $m$, the subsequence $\lambda^{(m)}|_{T_2}$ is consistent in the sense that if there is a repeated index $f(j_1) = f(j_2)$, then $(\lambda^{(m)})_{j_1} = (\lambda^{(m)})_{j_2}$. Therefore, $\lambda^{(m)}|_{T_2}$ can be uniquely determined by a function $f' \,:\, T_2 \rightarrow \{1,2,\ldots, M^{\alpha-1}\}$ specifying each molecule within $B_{i}$ (see Definition \ref{dfn:index_based}) when the index is $i = f(j)$ (with $j \in T_2$).  Since there are at most $R' M$ distinct values in $\{f(j)\}_{j \in T_2}$ (see Definition \ref{def:setS}) and each $B_i$ has size $\frac{M^{\alpha}}{M} = M^{\alpha - 1}$, there are at most $(M^{\alpha-1})^{R'M} = \exp(R'M(\alpha-1)\log M)$ such choices of $f'$, and thus at most $\exp(R'M(\alpha-1)\log M)$ possible values of $\lambda^{(m)}|_{T_2}$. 
\end{proof}

Next, define $V \subseteq U$ as follows:
\begin{equation}
    V = \{ m \in U_1 \,:\, \lambda^{(m)}|_{T_2} \text{ is unique within }U_1 \}. \label{eq:defV}
\end{equation}
By Lemma \ref{lem:unique_v}, $|V| \leq \exp(R'M(\alpha-1)\log M)$, so with a uniformly randomly message $m$ over $\exp(RML) = \exp(RM(\alpha-1)\log M)$ values (see \eqref{eq:num_messages}), we have
\begin{equation}
    \p(m \in V) \leq \frac{\exp(R' M (\alpha-1)\log M)}{\exp(R_0 M (\alpha-1)\log M)} = \exp(-(R_0-R') M (\alpha-1)\log M). \label{eq:V_prob}
\end{equation}
This property will be used shortly.

Using Lemma \ref{lem:unique_mt2}, we have for all $m \in U_1 \setminus V$ (i.e., messages in $U_1$ with non-unique $\lambda^{(m)}|_{T_2}$) that
\begin{equation}
        \p(\Lambdatil \in E_m \land \Psi=1|m) \geq p^{\delta M+1},
\end{equation}
which implies
\begin{equation}
    \p(\Lambdatil \in E_m\land \Psi=1 ) \geq \p(m \in U_1 \setminus V) p^{\delta M+1}.
    \label{eq:u_minus_v}
\end{equation}

We can also observe a useful fact about messages $m \notin U$: The adversary sets $\Psi = 0$ (and is thus inactive) with probability at least $1-p$, and when this happens, we have $|\Lambdatil|>c''M$ by the definition of $U$ in \eqref{eq:setU}.  Hence, we have for $m \notin U$ that
\begin{equation}
    \p(|\Lambdatil|\leq c''M|m) \leq p, \label{eq:not_in_U}
\end{equation}
which implies
\begin{equation}
    \p(|\Lambdatil| \leq c''M) \leq p \p(m \notin U) + \p(m \in U) \leq p + \p(m \in U_1) + \p(m\in U_2). \label{eq:combining_U}
\end{equation}
Re-arranging gives $\p(m \in U_1) \ge \p(|\Lambdatil| \leq c''M) - p - \p(m \in U_2)$, which we  combine with \eqref{eq:u_minus_v} as follows:
\begin{align}
    &\p(\Lambdatil \in E_m \land \Psi=1) \\
    &\geq \p(m \in U_1 \setminus V) p^{\delta M+1} \\
    &\geq (\p(m \in U_1) - \p(m \in V)) p^{\delta M+1} \\
    &\geq p^{\delta M+1}\Big( \p(|\Lambdatil| \leq c''M) - p - \p(m \in V) -\p(m \in U_2)\Big).
    \label{em_u1}
\end{align}
Observe that if $m\in U_2$ and $\Psi=0$, then the adversary does nothing and the decoder outputs something other than $m$, so an error occurs. Therefore,
\begin{equation}
    \p(\Lambdatil \in E_m \land \Psi=0) \geq (1-p) \p(m \in U_2).
    \label{em_u2}
\end{equation}
The overall error probability is then lower bounded by the sum of \eqref{em_u1} and \eqref{em_u2}.  Note that $1-p > p^{\delta M+1}$, which means the resulting coefficient of $\p(m \in U_2)$ is positive, and lower bounding this term by zero gives
\begin{equation}
    \p(\Lambdatil \in E_m) \geq p^{\delta M+1}\Big( \p(|\Lambdatil| \leq c''M) - p - \p(m \in V) \Big).
\end{equation}
Lemma \ref{lem:main} then follows by substituting the bound on $\p(m \in V)$ from \eqref{eq:V_prob}.

\subsection{Proof of Theorem \ref{thm:converse2} (Converse Under a Weak Adversary)} \label{sec:pf_weaker}

Under our weak adversary setup (Assumption \ref{def:adversary3}), the adversary can no longer look into the future in order to identify which $R' M$ (or more) molecules to leave unchanged and which $\delta M$ (or fewer) to change.  Accordingly, we adopt an adversary that \emph{randomly} chooses $R'M$ indices of molecules to be unchanged, and attempts to change certain molecules in the remaining indices to cause an error (if possible).  While this makes it less likely that the adversary will succeed in its goal, it turns out to suffice for attaining a converse having the same error exponent.

Formally, the adversary is described as follows:
\begin{itemize}
    \item Let $\mathcal{I}$ be a uniformly randomly chosen set of $R' M$ indices in $\{1,\dotsc,M\}$, meaning there are $\binom{M}{R'M}$ possible choices of $\mathcal{I}$.
    \item In accordance with Assumption \ref{def:adversary1}, the adversary knows true message $m$; recall that $A_m$ is the corresponding (outer) codeword.
    \item If there exists $m' \ne m$ such that $A_m|_{\mathcal{I}} = A_{m'}|_{\mathcal{I}}$, the adversary will choose such an $m'$ uniformly randomly.  Otherwise, we say that $m'$ is unspecified.
    \item As in Section \ref{sec:pf_prob_f}, the adversary additionally generates $\Psi \sim {\rm Bernoulli}(p)$ (independent of everything else).   If $m'$ is specified and $\Psi=1$ then we say that the adversary is \emph{active}, and otherwise the adversary is \emph{inactive} and leaves all molecules unchanged.
    \item When active, the adversary attempts to change every molecule to match the corresponding molecules of $A_{m'}$. The adversary will only achieve this goal if all of the required sequencing errors occur.
\end{itemize}

Let $t_m(f)$ be the stopping time given index sequence $f$ and message $m$ when there are no sequencing errors, and define
% Again, for any index sequence $f$, let $t_m(f)$ be the stopping time of molecule $m$ given index $f$, if there are no sequencing errors.
\begin{equation}
    U(f) = \{m : t_m(f) \leq c''M\}.
\end{equation}
%We will use the idea of messages $m \in U(f)$ being decoded by time $c'' M$ when the adversary is inactive (or equivalently, when there are no sequencing errors).  However, we will need to distinguish whether the resulting decoder output is indeed equal to $m$ or not.  Accordingly, 
Analogous to \eqref{eq:defU1}--\eqref{eq:defU2}, let $\Phi(f)$ be the set of messages such that the decoder outputs $m$ when the index sequence is $f$ and there are no sequencing errors, and partition $U(f)$ as follows:
\begin{gather}
    % U_1(f) = \{m | m \in U(f) \land \text{decoder outputs $m$ if the true molecule is $m$}\}
    U_1(f) = U(f) \cap \Phi(f), \\
    U_2(f) = U(f) \setminus \Phi(f) = U(f) \setminus U_1(f).
\end{gather}
Recall also the definition of $S(c'',R')$ in Definition \ref{def:setS}, and for each $f \in S(c'',R')$, let $T_1(f)$ and $T_2(f)$ be the sets meeting the conditions of Definition \ref{def:setS}, i.e., $|T_1(f)| \le \delta M$, and there are at most $R'M$ distinct values of $\{f(j)\}_{j \in T_2(f)}$. 

The following lemma follows easily from the preceding definitions.

\begin{lemma}
    If the message $m$ and the index sequence $f \in S(c'',R')$ satisfy all of the following conditions, then a decoding error occurs:
    \begin{itemize}
        \item The adversary's choice of $m'$ satisfies $m' \in U_1(f)$ (if no $m'$ satisfies $A_m|_{\mathcal{I}} = A_{m'}|_{\mathcal{I}}$, then $m'$ is unspecified and this condition is false); 
        \item $T_2(f) \subseteq \mathcal{I}$;
        \item For each index in $T_1(f)$, a sequencing error occurs at that index;
        \item The adversary's Bernoulli decision returns $\Psi=1$ (and thus the adversary is active).
    \end{itemize}
    \label{lem:misguess}
\end{lemma}
\begin{proof}
    Observe that under the above conditions, the adversary can modify all molecules to match the message $m'$, and since $m' \in U_1(f)$, the decoder stops before seeing $c''M$ molecules and declares that the message is $m'$.
\end{proof}

Next, we give a useful characterization of the probability of each possible $m'$ occurring.

\begin{lemma}
    Let $m_0$ be any message such that $A_{m_0}|_{\mathcal{I}}$ is not unique. Then the adversary's (random) choice of $m'$ satisfies $\p(m'=m_0) = \exp(-R_0M(\alpha-1)\log M)$.  (Note that we do not condition on the message $m$ here, but rather this probability is averaged over all possible values of $m$).
    \label{lem:possible_m'}
\end{lemma}
\begin{proof}
Let $\mathcal{M} = \mathcal{M}(m_0,\mathcal{I})$ be the set of all messages that share the same $A_{m_0}|_{\mathcal{I}}$. 
For each $m \in \mathcal{M}$ with $m \neq m_0$, the probability of the adversary choosing $m' = m_0$ conditioned on the true message being $m$ is exactly $\frac{1}{|\mathcal{M}|-1}$.  Recalling from \eqref{eq:num_messages} that there are $\exp(R_0 M(\alpha-1)\log M)$ messages in total, we also have $\p(m \in \mathcal{M} \land m \neq m_0) = \frac{|\mathcal{M}|-1}{\exp(R_0 M(\alpha-1)\log M)}$.  Combining the preceding two observations, we obtain  $\p(m'=m_0) = \exp(-R_0 M(\alpha-1)\log M)$ as desired.
\end{proof}

We can now proceed to compute the probability that the conditions of Lemma \ref{lem:misguess} hold, initially conditioning on fixed $f$ and $\mathcal{I}$. 
Define the set
\begin{equation}
    V(\mathcal{I}) = \{ m : A_m|_{\mathcal{I}} \text{ is unique}\}.
\end{equation}
By Lemma \ref{lem:possible_m'}, we have
\begin{align}
    \p(m' \in U_1(f) | f,\mathcal{I})
        \geq \frac{|U_1(f)\setminus V(\mathcal{I})|}{\exp(R_0M(\alpha-1)\log M)}. \label{eq:pm'_lb}
\end{align}
Note that if $m \notin U(f)$, then recalling that with probability $1-p$ the adversary generates $\Psi = 0$ and does nothing, we have $\p(|\Lambdatil|\leq c''M|m,f,\mathcal{I}) \leq p$. Hence,
\begin{align}
    &\p(|\Lambdatil| \leq c''M|f,\mathcal{I}) \nonumber \\ 
        &\leq \p(m \in U(f)|f,\mathcal{I}) + \sum_{m \notin U(f)} \p(m)\cdot p\\
        &\leq \p(m \in U_1(f)\setminus V(\mathcal{I})|f,\mathcal{I}) + \p(m \in V(\mathcal{I})|f,\mathcal{I}) + \p(m \in U_2(f)|f,\mathcal{I}) + p.
        \label{eq:size_y}
\end{align}
We evaluate the first term as follows:
\begin{align}
    \p(m \in U_1(f)\setminus V(\mathcal{I})|f,\mathcal{I}) &= \p(m \in U_1(f)\setminus V(\mathcal{I})) \\
    &= \frac{|U_1(f)\setminus V(\mathcal{I})|}{\exp(R_0M(\alpha-1)\log M)},
\end{align}
where the first step uses that $F$ and $\mathcal{I}$ are independent of $m$, and the second step uses the fact that $m$ is uniform over $\exp(R_0M(\alpha-1)\log M)$ messages.  Substituting into \eqref{eq:size_y} and re-arranging, we obtain
\begin{equation}
    \frac{|U_1(f)\setminus V(\mathcal{I})|}{\exp(R_0M(\alpha-1)\log M)} \ge \p(|\Lambdatil| \leq c''M|f,\mathcal{I}) - \p(m \in U_2(f)|f,\mathcal{I}) - \p(m \in V(\mathcal{I})|f,\mathcal{I}) - p, \label{eq:frac_lb}
\end{equation}
meaning the right-hand side lower bounds $\p(m' \in U_1(f) | f,\mathcal{I})$ via \eqref{eq:pm'_lb}.

Next, observe that each $A_m|_{\mathcal{I}}$ can be described by a function from $\mathcal{I}$ to the molecule in $A_m$ with the corresponding index in $\mathcal{I}$, so that there are at most $(M^{\alpha-1})^{R'M}$ such choices (similar to the proof of Lemma \ref{lem:unique_v}). Therefore,
\begin{equation}
|V(\mathcal{I})| \leq (\text{number of possible $A_m|_{\mathcal{I}}$}) \leq (M^{\alpha-1})^{R'M} = \exp(R'M(\alpha-1) \log M), \label{eq:V_size}
\end{equation}
which implies
\begin{equation}
\p(m \in V(\mathcal{I})|f,\mathcal{I}) \leq \frac{\exp(R'M(\alpha-1) \log M)}{\exp(R_0M(\alpha-1)\log M}. \label{eq:V_o1}
\end{equation}
We now check the conditions of Lemma \ref{lem:misguess}:
\begin{itemize}
    \item By \eqref{eq:pm'_lb}, the first condition of Lemma \ref{lem:misguess} holds with probability (conditioned on $f,\mathcal{I}$) at least
    \begin{equation}
        \frac{|U_1(f)\setminus V(\mathcal{I})|}{\exp(R_0M(\alpha-1)\log M)} \geq \p(|\Lambdatil|\leq c''M |f,\mathcal{I}) - \p(m \in U_2(f)|f,\mathcal{I}) - o(1),
    \end{equation}
    where the inequality follows by applying \eqref{eq:frac_lb} and noting that the last two terms therein are $o(1)$ due to \eqref{eq:V_o1} (with $R' < R_0$) and $p=o(1)$.
    % \eqref{eq:frac_lb} along with $|U(f) \setminus V(\mathcal{I})| \ge |U(f)| - |V(\mathcal{I})|$; the subtraction of $|V(\mathcal{I})|$ is asymptotically negligible since $|V(\mathcal{I})| \leq \exp(R'M(\alpha-1)\log M)$ (see \eqref{eq:V_size}) and $R' < R_0$.
    \item Conditioned on any $f \in S(c'', R')$, the second condition occurs with probability at least $2^{-M}$; this is because at least one such $\mathcal{I}$ exists, and the uniform distribution on $\mathcal{I}$ is over at most $2^M$ elements.
    \item Conditioned on any $f\in S(c'', R')$, the third condition occurs with probability at least $p^{\delta M}$, since $|T_1(f)| \le \delta M$.
    \item Lastly, the adversary sets $\Psi=1$ with probability $p$.
\end{itemize}
Combining these and averaging over $\mathcal{I}$, the error probability conditioned on any fixed $f \in S(c'',R')$ is lower bounded as follows:
\begin{equation}
    \p(\text{error} \wedge \Psi=1|f) \geq (\p(|\Lambdatil|\leq c''M |f) -\p(m \in U_2(f)|f)- o(1))\cdot 2^{-M} \cdot p^{\delta M} \cdot p,
\end{equation}
where we used the fact that the relevant sources of randomness (for $m'$, $\mathcal{I}$, sequencing errors, and $\Psi$) are independent.

If the adversary is not active, then an error will occur whenever $m \in U_2(f)$ by the definition of $U_2$. Therefore,
\begin{equation}
    \p(\text{error} \wedge \Psi=0|f) \geq \p(m \in U_2(f) | f)(1-p).
\end{equation}
Adding together the terms for $\Psi=0$ and $\Psi=1$, we see that the overall coefficient to $\p(m \in U_2(f)|f)$ is positive, and lower bounding this term by zero gives
\begin{align}
    % & \p(\text{error} | f) \nonumber \\
    % &\geq  (\p(|\Lambdatil|\leq c''M |f) -\p(m \in U_2(f)|f)- o(1))\cdot 2^{-M} \cdot p^{\delta M} \cdot p + (1-p)\p(m \in U_2(f)|f) \label{eq:pre_two_cases}\\
    % & 
    \p(\text{error} | f)  \geq (\p(|\Lambdatil|\leq c''M |f)- o(1))\cdot 2^{-M} \cdot p^{\delta M}. \label{eq:two_cases}
\end{align}
% where \eqref{eq:two_cases} follows by applying $1-p > 2^{-M}p^{\delta M}p$ and $\p(m \in U_2(f)|f) \ge 0$.
% follows by considering two cases:
% \begin{itemize}
%     \item If $\p(m \in U_2(f)|f) \ge \p(|\Lambdatil|\leq c''M |f) \cdot p$, then the second term in \eqref{eq:pre_two_cases} is much larger than \eqref{eq:two_cases}.
%     \item If $\p(m \in U_2(f)|f) < \p(|\Lambdatil|\leq c''M |f) \cdot p$, then the first term in \eqref{eq:pre_two_cases} simplifies to \eqref{eq:two_cases} by applying $p = o(1)$.
% \end{itemize}
Averaging \eqref{eq:two_cases} over $f$, it follows that
\begin{align}
    \p(\text{error})
    &\geq \sum_{f \in S(c'',R')} \p(f) \cdot \p(\Lambdatil \in E_m | f)\\
    & \geq \sum_{f\in S(c'',R')} \p(f) (\p(|\Lambdatil|\leq c''M |f) - o(1))\cdot 2^{-M} \cdot p^{\delta M+1}.
\end{align}
This result matches \eqref{eq:lb_Em_final} from the proof of Theorem \ref{thm:converse}, but with an extra $2^{-M}$ factor.  The remainder of the proof is identical to that of Theorem \ref{thm:converse} but with this extra factor incorporated, and thus the final scaling in the theorem statement is $o(2^{-M} p^{\delta M+1})$ instad of $o(p^{\delta M + 1})$.

\section{Conclusion}

We have established achievablity and converse results for error exponents of index-based concatenated codes under a variable number of reads.  Our achievability result illustrates a significant reduction in error probability compared to the fixed-reads setting (from $e^{-\Theta(M)}$ to $e^{-\Theta(M\log\frac{1}{p})}$), and our converse shows result that our achievability result is tight in certain regimes where a relatively small number of reads is performed.

Perhaps the most immediate direction for future research is to establish converse bounds and/or improved achievability bounds that are suited to a higher average number of reads.  As exemplified in Figure \ref{fig:examples}, we do not provide a converse in such regimes, and our achievable exponent approaches a finite limit no matter how many reads are done.  (This should not be unavoidable, since with infinitely many reads the error probability is zero, corresponding to an infinite error exponent.)  Another interesting direction would be to study different models of sequencing errors, particularly random instead of adversarial.  

\bibliographystyle{IEEEtran}
\bibliography{general}

\end{document}